\documentclass[11pt]{article}
\usepackage{amsmath,amssymb,amsfonts,amsthm,appendix}
\usepackage[citecolor=blue,linkcolor=blue,colorlinks=true]{hyperref}

\usepackage{complexity}
\usepackage[capitalise] {cleveref}	
\usepackage{xspace}

\usepackage{todonotes}

\newcounter{todocounter}

\theoremstyle{plain}
\newtheorem{theorem}{Theorem}[section]
\newtheorem{lem}[theorem]{Lemma} 
\newtheorem{proposition}[theorem]{Proposition}
\newtheorem{claim}[theorem]{Claim}
\newtheorem{lemma}[theorem]{Lemma}
\newtheorem{corollary}[theorem]{Corollary}

\newtheorem{observation}[theorem]{Observation}
\newtheorem{question}[theorem]{Question}
\newtheorem{rem}[theorem]{Remark}

\theoremstyle{definition}
\newtheorem{defnn}[theorem]{Definition}

\newcommand{\sse}{\subseteq}
\newcommand{\bits}{\{-1,1\}}

\newcommand{\zon}{\{0,1\}^n}
\newcommand{\defn}{\stackrel{\text{\tiny def}}{=}}

\newcommand{\zo}{\bits}

\newcommand{\pmo}{\set{-1,1}}

\newcommand{\sps}{\mathsf{sparsity}}

\newcommand{\set}[1]{\left\{ #1 \right\}}
\newcommand{\etal}{\textit{et al}.\@\xspace}
\newcommand{\ie}{i.e.}

\newcommand{\DT}{\mathsf{DT}}
\newcommand{\rand}{\mathsf{R}}

\newcommand{\sens}{\mathsf{s}}
\newcommand{\bsens}{\mathsf{bs}}
\newcommand{\cert}{\mathsf{C}}
\newcommand{\alt}{\mathsf{alt}}
\newcommand{\salt}{\mathsf{salt}}
\newcommand{\scalt}{\mathsf{scalt}}
\renewcommand{\deg}{\mathsf{deg}}

\newcommand{\calC}{{\cal C}}
\newcommand{\calF}{{\cal F}}

\newcommand{\INP}{\mathsf{IP}}

\newcommand{\GIP}{\mathsf{GIP}}

\newcommand{\maj}{\mathsf{Maj}}
\newcommand{\smaj}{\mathsf{ShiftMaj}}

\newcommand{\F}{{\mathbb{F}}}
\newcommand{\N}{{\mathbb{N}}}

\newcommand{\degtwo}{\mathsf{deg_{2}}}
\renewcommand{\deg}{\mathsf{deg}}

\usepackage{collect,enumerate,fullpage}
\title{Sensitivity, Affine Transforms and \\ Quantum Communication Complexity}

\author{Krishnamoorthy Dinesh\thanks{Indian Institute of Technology Madras, Chennai, India. {\small \texttt{\{kdinesh,jayalal\}@cse.iitm.ac.in}}} \and  Jayalal Sarma$^*$}
\date{}

\makeatletter
\newtheorem*{rep@theorem}{\rep@title} \newcommand{\newreptheorem}[2]{
\newenvironment{rep#1}[1]{
	\def\rep@title{\textsf{\textbf{#2}} \ref{##1}}
	\begin{rep@theorem} }
	{\end{rep@theorem} } }
	\makeatother

\newreptheorem{theorem}{Theorem}
\newreptheorem{lemma}{Lemma}
\newreptheorem{prop}{Proposition}

\begin{document}

\maketitle

\begin{abstract}
In this paper, we study the Boolean function parameters sensitivity ($\sens$), block sensitivity ($\bsens$), and alternation ($\alt$) 
under specially designed affine transforms and show several applications. For a 
function $f:\F_2^n \to \zo$, and $A = Mx+b$ for $M \in \F_2^{n \times n}$ and $b \in 
\F_2^n$, the result of the transformation $g$ is defined as $\forall x \in \F_2^n, g(x) = 
f(Mx+b)$. 

As a warm up, we study alternation under linear shifts (when $M$ is restricted to be the identity matrix) called the \emph{shift invariant alternation} (the smallest alternation that can be achieved for the Boolean function $f$ by shifts, denoted by $\salt(f)$). By a result of Lin and Zhang [ICALP 2017], it follows that $\bsens(f) \le O(\salt(f)^2\sens(f))$. Thus, to settle the Sensitivity Conjecture ($\forall~f, \bsens(f) \le \poly(\sens(f))$), it suffices to argue that $\forall~f, \salt(f) \le \poly(\sens(f))$. However, we exhibit an explicit family of Boolean functions for which $\salt(f)$ is $2^{\Omega(\sens(f))}$.

Going further, we use an affine transform $A$, such that the corresponding 
function $g$ satisfies $\bsens(f,0^n) \le \sens(g)$. We apply this 
in the setting of quantum communication complexity 
to  prove that for $F(x,y) \defn f(x \land y)$, the bounded error quantum communication complexity of $F$ with prior 
entanglement, $Q^*_{1/3}(F)$ is $\Omega(\sqrt{\bsens(f,0^n)})$. 
Our proof builds on ideas from Sherstov [Quantum Information and Computation, 10:435–455, 2010]
where we use specific properties of the above affine transformation.
Using this, we show the following.
\begin{enumerate}[(a)]

\item For a fixed prime $p$ and an $\epsilon$, $0 < \epsilon < 1$, any Boolean function $f$ that depends on all its inputs with $\deg_p(f) \le (1-\epsilon)\log n$ must satisfy $Q^*_{1/3}(F) = \Omega\left (\frac{n^{\epsilon/2}}{\log n} \right )$. Here, $\deg_p(f)$ denotes the degree of the multilinear polynomial over $\F_p$ which agrees with $f$ on Boolean inputs.

\item For Boolean function $f$ such that there exists primes $p$ and $q$ with  $\deg_q(f) \ge \Omega(\deg_p(f)^\delta)$ for $\delta > 2$,  
the deterministic communication complexity - $\D(F)$ and $Q^*_{1/3}(F)$ are polynomially related. In particular, this holds when $\deg_p(f) = O(1)$. 
Thus, for this class of functions, this answers an open question (see Buhrman and de Wolf [CCC 2001]) about the relation between the two measures.
\end{enumerate}

\vspace{-1mm}
\noindent Restricting back to the linear setting, we construct linear transformation $A$, such that 
the corresponding function $g$ satisfies, $\alt(f) \le 2\sens(g)+1$. Using this new 
relation, we exhibit Boolean functions $f$ (other than the parity function) such that $
\sens(f)$ is $\Omega(\sqrt{\sps(f)})$ where $\sps(f)$ is the number of non-zero coefficients in the Fourier representation of $f$. This family of Boolean functions also rule out a potential approach to settle the XOR Log-Rank conjecture via the recently settled Sensitivity conjecture [Hao Huang, Annals of Mathematics, 190(3): 949-955, 2019].

\end{abstract}
\section{Introduction}
For a Boolean function $f:\zon \to \zo$, \emph{sensitivity} of $f$ on $x \in \zon$, is 
the maximum number of indices $i \in [n]$, such that $f(x \oplus e_i) \ne f(x)$ where 
$e_i \in \zon$ with exactly the $i^{th}$ bit as $1$. The \textit{sensitivity} of $f$ (denoted by $
\sens(f)$) is the maximum sensitivity of $f$ over all inputs. A related parameter is the 
\emph{block sensitivity} of $f$ (denoted by $\bsens(f)$), where we allow disjoint blocks 
of indices to be flipped instead of a single bit. Another parameter is the deterministic \emph{decision tree complexity} (denoted by $\DT(f)$) which is the depth of an optimal decision tree computing the function $f$. The \emph{certificate complexity} of $f$ (denoted by $\cert(f)$) is the non-deterministic variant of the decision tree complexity. The parameter $\sens(f)$ was originally studied by Cook \etal~\cite{CDR86} in connection with the {\sc CREW-PRAM} model of computation. Subsequently, Nisan and Szegedy~\cite{NS94} (see also~\cite{N91}) introduced the parameters $\bsens(f)$ and $\cert(f)$ and conjectured that for any function $f:\zon \to \zo$, $\bsens(f) \le \poly(\sens(f))$ - known as the Sensitivity Conjecture. Later 
developments, which revealed several connections between sensitivity, block sensitivity 
and the other Boolean function parameters, demonstrated the fundamental nature of the 
conjecture (see \cite{HKP11} for a survey and several equivalent formulations of the 
conjecture). This conjecture has recently been resolved in~\cite{H19} by showing the following which implies that $\bsens(f) = O(\sens(f)^4)$.
\begin{theorem}[Sensitivity Theorem~\cite{H19}] \label{sens:thm}
	For every Boolean function $f$, $\deg(f) \le \sens(f)^2$. 
\end{theorem}

Shi 
and Zhang~\cite{ZS10} studied the parity complexity variants of $\bsens(f), \cert(f)$ and 
$\DT(f)$ and observed that such variants have the property that they 
are invariant under arbitrary invertible linear transforms (over $\F_2^n$). They also showed existence of Boolean functions where under 
\emph{all} invertible linear transforms of the function, the decision tree depth is linear while their parity variant of decision tree complexity is at most logarithmic in the input length.

\noindent
\textbf{Our Results :} 
While the existing studies focus on understanding the Boolean function parameters under 
the effect of arbitrary invertible affine transforms, 
in this work, we study the relationship between the above parameters of Boolean functions 
$f : \F_2^n \to \zo$, under specific affine transformations over $\F_2^n$. More 
precisely, we explore the relationship of the above parameters for the function $g : 
\F_2^n \to \zo$ and $f$, where $g$ is defined as $g(x) = f(Mx+b)$ for specific $M \in \F_2^{n 
\times n}$ and $b \in \F_2^n$ (where is $M$ not necessarily invertible). We show the following results, and their corresponding 
applications, which we explain along with the context in which they are relevant.

\noindent
\textbf{Alternation under shifts :} We study the parameters when 
the transformation is very structured - namely the matrix $M$ is the identity matrix and 
$b \in \F_2^n$ is a linear shift. More precisely, we study $f_b(x) \defn f(x+b)$ where 
$b$ is the shift. 
Observe that all the parameters mentioned above are invariant under shifts. A Boolean function parameter which is neither shift invariant nor invariant 
under invertible linear transforms is the \emph{alternation}, a measure of 
non-monotonicity of Boolean function (see~\cref{sec:prelims} for a formal definition). To see 
this for the case of shifts, if we take $f$ as the majority function on $n$ bits, then there exists shifts $b \in \zon$ where $\alt(f_b) = \Omega(n)$ while $\alt(f) = 1$.

A result related to Sensitivity Conjecture by Lin and Zhang~\cite{LZ17} shows that 
$\bsens(f) \le O(\sens(f)\alt(f)^2)$.
This bound for $\bsens(f)$, implies that to settle the Sensitivity Conjecture, it suffices to show 
that $\alt(f)$ is upper bounded by $\poly(\sens(f))$ for all Boolean functions $f$. However, the authors~\cite{JDS18} ruled this out, by exhibiting a family of functions where $\alt(f)$ is at least 
$2^{\Omega(\sens(f))}$. 

Observing that the parameters $\sens(f), \bsens(f)$ are invariant under shifts, we 
define a new quantity \emph{shift-invariant alternation}, $\salt(f),$ which is the minimum alternation of any function $g$ 
obtained from $f$ upon shifting by a vector $b \in \zon$ (\cref{def:salt}). By the aforementioned bound on $\bsens(f)$ of~\cite{LZ17}, it is easy to observe that $\bsens(f) \le O(\sens(f)\salt(f)^2 )$. We also show that there exists a family of Boolean functions $f$ with $\bsens(f) = \Omega(\sens(f)\salt(f))$ (\cref{prop:tightness-salt}). 

It is conceivable that $\salt(f)$ is 
much smaller compared to $\alt(f)$ for a Boolean function $f$ and hence that $\salt(f)$ can potentially be upper 
bounded by $\poly(\sens(f))$ thereby settling the Sensitivity Conjecture. However, we rule this out by showing the following 
stronger gap, about the same family of functions demonstrated in \cite{JDS18} (see also \cite{GSW16}).
\begin{proposition}\label{salt:gap}
There exists an explicit family of Boolean functions for which $\salt(f)$ is $2^{\Omega(\sens(f))}$.
\end{proposition}

\noindent
\textbf{Block Sensitivity under Affine Transformations :}
We now generalize our theme of study to the affine transforms over $\F_2^n$. In 
particular, we explore how to design affine transformations in such a way that block sensitivity of the original function ($f$) is upper bounded by the sensitivity of the new function ($g$). We use $\bsens(f,a)$ to denote the number of sensitive blocks of $f$ on the input $a$.
\begin{lem}\label{lt:bs-s}
For any $f:\F_2^n \to \pmo$ and $a \in \zon$, there exists an affine transform $A:
\F_2^n 	\to \F_2^n$	such that for $g(x) = f(A(x))$, 
\begin{enumerate}[(a)]
\item $\bsens(f,a)\le \sens(g,0^n)$, and 
\item $g(x) = f((x_{i_1}, x_{i_2}, \ldots, x_{i_n})\oplus a)$ where $i_1, \ldots, 
i_n \in [n]$ are not necessarily distinct.
\end{enumerate}
\end{lem}

The above transformation is used in  Nisan and Szegedy (see Lemma 7 of~\cite{NS94}) 
to show that $\bsens(f) \le 2\deg(f)^2$. Here, $\deg(f)$ is the degree of the 
multilinear polynomial over reals that agrees with $f$ on Boolean inputs.
 We show another application of~
\cref{lt:bs-s} in the context of quantum communication complexity, a model for 
which 
was introduced by Yao~\cite{Y93}. In this model, two parties Alice and Bob have to 
compute a function $F:\zon \times \zon \to \zo$, where Alice is given an $x \in 
\zon$ and Bob is given a $y \in \zon$. Both the parties have to come up with a 
\emph{quantum protocol} where they communicate qubits via a quantum channel and 
compute $f$ while minimizing the number of qubits exchanged (which is the 
\emph{cost} of the quantum protocol) in the process. In this model, we allow 
protocols to have prior entanglement. We define $Q^*_{1/3}(F)$ as the minimum cost 
quantum protocol computing $F$ with prior entanglement. For more details on this
model, see~\cite{R03}. The corresponding analog in the classical setting is the 
bounded error randomized communication model where the parties communicate with
$0,1$ bits and share an unbiased random source. We define $\rand_{1/3}(F)$ as the 
minimum cost randomized protocol computing $F$ with error at most $1/3$. It can be 
shown that $Q^*_{1/3}(F) \le \rand_{1/3} (F) \le \D(F)$.

One of the fundamental goals in quantum communication complexity is to see if there are 
functions where their randomized communication complexity is significantly larger than their quantum communication complexity. It has been the conjectured by Shi and Zhu~\cite{SZ09} that this is not the case in general (which they called the Log-Equivalence Conjecture). In this work, we are interested in the case when $F(x,y)$ is of the form $f(x \land y)$ where $f:\zon \to \zo$ and $x \land y$ is the string obtained by bitwise AND of $x$ and $y$.
\begin{question}
\label{qc-equiv}
For $f:\zon \to \zo$, let $F: \zon \times \zon \to \zo$ be defined as $F(x,y) = f(x 
\land y)$. Is it true that for any such $F$, $\D(F) \le \poly(Q^*_{1/3}(F))$ ?
\end{question}

Since $\rand_{1/3}(F) \le \D(F)$, answering the above question in positive would show that the
classical randomized communication model is as powerful as the quantum communication model 
for the class of functions $F(x,y) = f(x \land y)$. This question for such restricted $F$ 
has also been proposed by Klauck~\cite{K07} as a first step towards answering the general 
question (see also~\cite{BW01}). In this direction, Razborov~\cite{R03} showed that for the special case when $f
$ is symmetric, $F(x,y) = f(x \land y)$ satisfy $\D(F) \le O(Q^*_{1/3}(F)^2)$. In the process, Razborov developed powerful techniques to obtain lower bounds on $Q^*_{1/3}(F)$ which were subsequently generalized by Sherstov~\cite{S08}, Shi and Zhu~\cite{SZ09}. Subsequently, in a slightly different direction, Sherstov~\cite{S10} showed that instead of computing $F(x,y) = f(x \land y)
$ alone, if we consider $F$ to be the problem of computing both of $F_1(x,y) = f(x \land y)$ and 
$F_2(x,y) = f(x \lor y)$, then $\D(F) = O(Q^*_{1/3}(F)^{12})$ for all Boolean functions $f$ where $Q^*_{1/3}(F) = \max\set{Q^*_{1/3}(F_1), Q^*_{1/3}(F_2)}$ and $\D(F) =\max\set{\D(F_1), \D(F_2)}$. Using~\cref{lt:bs-s}, we build 
on the ideas of Sherstov~\cite{S10} and obtain a lower bound for $Q^*_{1/3}(F)$ where $F(x,y) = F_1(x,y) = f(x \land y)$.
\begin{theorem}\label{ub:quantum:bs}
Let $f:\zon \to \pmo$ and $F(x,y) = f(x \land y)$, then,
$$Q^*_{1/3}(F) = \Omega\left(\sqrt{\bsens(f,0^n)}\right).$$
\end{theorem}
In this context, we make an important comparison\footnote{Recently, it was noticed that \cref{ub:quantum:bs} had already appeared in arXiv version 1 of \cite{S10arxiv} but did not appear in later versions.}
with a result of Sherstov~\cite{S10}. He proved that  for  
$F'(x,y) = f_b(x \land y) $, where $b \in \zon$ is the input on which $\bsens(f,x)
$ is maximum, $Q^*_{1/3}(F') = \Omega(\sqrt{\bsens(f)}) \ge 
\Omega(\sqrt{\bsens(f,0^n)})$ (Corollary 4.5 of~\cite{S10}).  Notice that $F$ and $F'$ differ by a linear shift 
of $f$ with $b$.\footnote{More importantly, this $b$ in Corollary 4.5 of~\cite{S10} cannot be fixed to $0^n$ for all Boolean functions  to conclude~\cref{ub:quantum:bs}. See~\cref{app:sherstov} for details.}  Moreover, $Q^*_{1/3}(F)$ can change drastically even under such 
(special) linear shifts of $f$. For example, consider $f=\land_n$. Since $
\bsens(f)$ is maximized at $1^n$, $b=1^n$. Hence, the function $F'$ is the 
disjointness function for which $Q^*_{1/3}(F') = \Omega(\sqrt{n})$~
\cite{R03} whereas, $Q^*_{1/3}(F) = O(1)$. The same counterexample also shows that 
$Q^*_{1/3}(F) = \Omega(\sqrt{\bsens(f)})$ cannot hold for all $f$ 
(see~\cref{rem:quantum}). Since the lower bounds shown on quantum communication complexity are on different functions, \cref{ub:quantum:bs} is incomparable with the result of Sherstov (Corollary 4.5 of~\cite{S10}).

Using the above result, for a prime $p$, we show that if $f$ has small degree when expressed as a polynomial over $\F_p$ (denoted by $\deg_p(f)$), the quantum communication complexity of $F$ is large. 
\begin{theorem}\label{thm:quant:lb}
Fix a prime $p$. Let $f:\zon \to \pmo$ where $f$ depends on all the variables. Let $F(x,y) = f(x \land y)$. For any $0 < \epsilon < 1$  such that $\deg_p(f) \le (1-\epsilon)\log n$, we have  $$Q^*_{1/3}(F) = \Omega\left (\frac{n^{\epsilon/2}}{\log n} \right ).$$ 
\end{theorem}

Observe that, though~\cref{ub:quantum:bs} does not answer~\cref{qc-equiv} in positive for all functions, we could show a class of Boolean function for which $\D(F)$ and $Q^*_{1/3}(F)$ are polynomially related. More specifically, we show this for the set of all Boolean functions $f$ such that there exists two distinct primes $p,q$ with $\deg_p(f)$ and $\deg_q(f)$ are sufficiently far apart (\cref{thm:qc-eq:spl}).
\begin{theorem}\label{thm:qc-eq:spl}
Let $f:\zon \to \pmo$ with $F(x,y) = f(x \land y)$. 
Fix $0 < \epsilon < 1$. If there exists distinct primes $p$, $q$ such that $\deg_q(f) = \Omega ( \deg_p(f)^{\frac{2}{1-\epsilon}})$, then $\D(F) = O(Q^*_{1/3}(F)^{2/\epsilon})$.
\end{theorem}
 By the result of Gopalan \etal (Theorem 1.2, \cite{GLS09}), any Boolean function $f$ with $\deg_p(f) = o(\log n)$ must have $\deg_q(f) = \Omega(n^{1-o(1)})$ thereby satisfying the condition of ~\cref{thm:qc-eq:spl}. Hence for all such functions, \cref{thm:qc-eq:spl} answers~\cref{qc-equiv} in positive. Observe that the same can also be derived from~\cref{thm:quant:lb}.

\noindent
\textbf{Alternation under Linear Transforms :} We now restrict our study to linear 
transforms. Again, in this context, the aim is to design special linear transforms for the parameters of interest.   In 
particular, in this case, we show linear transforms for which we can upper bound 
the alternation of the original function in terms of the sensitivity of the 
resulting function. More precisely, we prove the following lemma:
\begin{lem}\label{alt-sens-lt}
	For any $f:\F_2^n \to \zo$, there exists an invertible linear transform $L:
	\F_2^n 	\to \F_2^n$	such that for $g(x) = f(L(x))$, $$\alt(f) \le 2\sens(g)+1.$$ 
\end{lem}

We show an application of the above result in the context of the parameter sensitivity.
Nisan and Szegedy~\cite{NS94} showed that for any Boolean function $f$, $\sens(f) \le 2\deg(f)^2$. However, the situation is quite different for $\degtwo(f)$ - noticing that 
for $f$ being parity on $n$ variables, $\degtwo(f)=1$ and $\sens(f)=n$ - the gap can even 
be unbounded. 
Though parity may appear as a  corner case, there are other functions like the Boolean inner product function\footnote{$\INP_n(x_1,x_2,\ldots,x_n,y_1,y_2,\ldots,y_n) = \sum_i x_iy_i \mod 2$} $\IP_n$ 
whose $\F_2$-degree is constant while sensitivity is $\Omega(n)$ thereby ruling out the possibility that $\sens(f) \le \degtwo(f)^2$. 
It is known that if $f$ is not the parity on $n$ variables (or its negation), $\degtwo(f) \le \log \sps(f)$~\cite{BC99,GOSSW11}. Hence, as a structural question about the two parameters, we ask : for $f$ other than the parity function, is it true that $\sens(f) \le \poly(\log \sps(f))$.\footnote{Observe that functions like $\IP_n$ though have low $\F_2$-degree similar to parity however have high sparsity and hence does not rule this out.} In fact, the Sensitivity Theorem (\cref{sens:thm}) by~\cite{H19} implies that for every Boolean function $f$, $\log \sps(f) = O(\sens(f)^2)$. Hence, if we could answer our question in affirmative, it would imply that $\sens(f)$ and $\log \sps(f)$ are polynomially related. We use~\cref{alt-sens-lt}, which is in 
the theme of studying alternation and sensitivity in the context of linear 
transformations, to show that this is not the case, by exhibiting a family of functions where the gap is exponential.

\begin{theorem}\label{gap:salt:degtwo}
There exists a family of functions $\set{g_k\mid k \in \N}$ such that $$\sens(g_k) 
\ge \frac{\sqrt{\sps(g_k)}}{2}-1.$$
\end{theorem}

This family of Boolean functions also rules out a potential approach to settle the XOR Log-Rank conjecture via the recently settled Sensitivity conjecture~\cite{H19}. We elaborate on this approach and how our function family rules it out in~\cref{sec:sens-sps}. 
\section{Preliminaries}\label{sec:prelims}
In this section, we define the notations used. Define $[n] = \set{1,2,\ldots,n}$. 
For $S \sse [n]$, define $e_S \in \zon$  to be the indicator vector of the set $S$. 
For $x,y \in \zon$, we denote $x \land y$ (resp. $x\oplus y$) $\in \zon$ as the 
string obtained by bitwise AND (resp. XOR) of $x$ and $y$. We use $x_i$ to denote the $i^{th}$ bit of $x$.

We now define the Boolean function parameters we use. Let $f:\zon \to \zo$ and $a 
\in \zon$, we define, 1) the \emph{sensitivity} of $f$ on $a$ as $\sens(f,a) = |
\set{i \mid  f(a \oplus e_i) \ne f(a), i \in [n] }|$, 2) the \emph{block 
sensitivity} of $f$ on $a$, $\bsens(f,a)$ to be the maximum number of disjoint 
blocks $\set{B_i \mid B_i \sse [n]}$ such that $f(a \oplus e_{B_i}) \ne f(a)$ and 
3) the \emph{certificate complexity} of $f$ on $a$, $\cert(f,a)$ to be the size of 
the smallest set $S \sse [n]$ such that fixing $f$ according to $a$ on the location 
indexed by $S$ causes the function to become constant. For $\phi \in \set{\sens, 
\bsens, \cert}$, we define $\phi(f) = \max_{a \in \zon} \phi(f,a)$ and are 
respectively called the sensitivity, the block sensitivity and the certificate 
complexity of $f$. By definition, the three parameters are shift invariant, by 
which we mean $\forall~b \in \zon$, $\phi(f_b) = \phi(f)$ for $\phi \in \set{\sens, 
\bsens,\cert}$ where $f_b(x) \defn f(x \oplus b)$. Also, it can be shown that $\sens(f) \le \bsens(f) \le \cert(f)$.

For $x,y \in \zon$, define $x \prec y$ if $\forall i \in [n]$, $x_i \le y_i$. We 
define a \emph{chain} $\calC $ on $\zon$ as $(0^n = x^{(0)},x^{(1)},\ldots, x^{(n-1)},x^{(n)}=1^n)$ such that for all $i \in [n]$, $x^{(i)} \in \zon$ and 
$x^{(i-1)} \prec x^{(i)}$ . We define \emph{alternation} of $f$ for a chain $
\calC$, denoted $\alt(f,\calC)$ as the number of times the value of $f$ changes in the 
chain. We define alternation of a function $\alt(f)$ as $\max_{\text{ chain }\calC} 
\alt(f,\calC)$.

Every Boolean function $f$ can be expressed uniquely as a multilinear polynomial $p(x)$ in 
$\F[x_1,\ldots, x_n]$ over any field $\F$ such that $p(x) = f(x)~\forall x \in \zon$. Fix a prime $p$. We 
denote $\deg(f)$ (resp. $\deg_p(f)$) to be the degree of the multilinear polynomial 
computing $f$ over reals (resp. $\F_p$). We define $\DT(f)$ as the depth of an optimal decision 
tree computing $f$. It is known that for all Boolean functions $f$, $\deg_p(f) \le 
\deg(f) \le \DT(f) \le \bsens(f)^3$. 

Sparsity of a Boolean function $f:\zon \to \pmo$ (denoted by $\sps(f)$) is the number of non-zero Fourier coefficients in the Fourier representation of $f$. For more details on this parameter, see~\cite{AOBF}. For more details on $\DT(f)$ and other related parameters, 
see the survey by Buhrman, de Wolf~\cite{BW02} and Hatami \etal~\cite{HKP11}.

We consider the two party classical communication model. Given a function $f:
\zon \times \zon \to \zo$, Alice is given an $x \in \zon$ and Bob is given $y 
\in \zon$. They can communicate with each other and their aim is to compute 
$f(x,y)$ while communicating minimum number of bits. We call the procedure 
employed by Alice and Bob to computing $f$ as the \emph{protocol}. We define $
\D(f)$ as the minimum cost of a deterministic protocol computing $f$. For 
functions of the form $F(x,y) = f(x\land y)$, it is known that $\D(F) \le 
2\DT(f)$~\cite{MO09}. For more details on communication complexity of Boolean 
functions, refer~\cite{textbook}.

\section{Warm up: Alternation under Shifts} \label{sec:salt}

In this section, as a warm-up, we study sensitivity and alternation under linear shifts (when the matrix $M$ is the identity matrix).
We introduce a parameter, \textit{shift-invariant alternation} ($\salt$). We then show the existence of Boolean functions whose shift-invariant alternation is exponential in its sensitivity (see \cref{salt:gap}) thereby ruling out the possibility that $\salt(f)$ can be upper bounded by a polynomial in $\sens(f)$ for all Boolean functions $f$.

Recall from the introduction that the parameters $\sens,\bsens$ and $\cert$ are shift invariant while $\alt$ 
is not. To see that $\alt$ is not shift-invariant, for an even number $n$, consider the Boolean function defined as $\maj_n(x) = 1 \iff \sum_i x_i > n/2 $. For an even $n$, define $\smaj_n(x) = \maj_n(x\oplus 1^{n/2}0^{n/2})$. It is possible to exhibit a chain $\sigma$ such that $\alt(\smaj_n,\sigma) = n$, while $\alt(\smaj_n(x\oplus 1^{n/2}0^{n/2})) = \alt(\maj_n)= 1$. 

We define a variant of alternation which is invariant under shifts.
\begin{defnn}[\textbf{Shift-invariant Alternation}]
\label{def:salt}
For $f : \zon \to \zo$, the \emph{shift-invariant alternation} (denoted by $\salt(f)$) is defined as $\min_{b \in \zon} \alt(f_b)$. 
\end{defnn}
We remark that $\salt(\smaj_n)=1$. Hence the gap between measures $\alt$ and $\salt$ can be unbounded.

{\noindent \textbf{A family of functions with $\salt(f) = \Omega(2^{\sens(f)})$ :}
We now exhibit a family of functions $\calF$ where for all $f \in \calF$, $\salt(f) \ge 2^{\sens(f)}$ thereby ruling out the possibility that $\salt(f)$ can be upper bounded by a polynomial in $\sens(f)$. The family $\calF$ is the same class of Boolean functions for which alternation is at least exponential in sensitivity due to~\cite{JDS18}.

\begin{defnn}[Definition 1 from~\cite{JDS18}. See also Proof of Lemma A.1 of~
\cite{GSW16}]\label{def:fun}
Consider the family defined as follows. $$\calF = \set{f_k \mid f_k: \{0,1\}^{2^k-1}  \to \zo,  k \in \N}$$ 
 The Boolean function $f_k$ is computed by a decision tree which is a 
full binary tree of depth $k$ with $2^k$ leaves. A leaf node is labeled as $0$ (resp. 
$1$) if it is the left (resp. right) child of its parent. All the nodes (except the 
leaves) are labeled by a distinct variable. 
\end{defnn}
We remark that Gopalan \etal~\cite{GSW16} demonstrates an exponential lower bound on tree sensitivity (introduced by them as a generalization of the 
parameter sensitivity) in terms of decision tree depth for the same family of functions in~\cref{def:fun}. We remark that, in general,  lower bound on tree sensitivity need not implies a lower bound on alternation. For instance, if we consider the Majority function $\maj_n$, the tree sensitivity can be shown to be $\Omega(n)$ while alternation is $1$. 

The authors~\cite{JDS18} have shown that for any $f \in \calF$,  there exists of a chain of large alternation in $f$. However, this is not sufficient to argue existence of a chain of large alternation under \textit{every} linear shift.
We now proceed to prove an exponential lower bound on $\salt(f)$ in terms of $s(f)$ for all $ f\in \calF$. 
\begin{repprop}{salt:gap}
For $f_k \in \calF$,  $\salt(f_k) \ge 2^{\Omega(\sens(f_k))}$. 
\end{repprop}
\begin{proof}
We show\footnote{In this proof, for simplicity, we abuse the notation $f_k(x \oplus c)$ to denote the function obtained by shifting $f_k$ by $c$.} that for $f_k \in \calF$ and $n = 2^k-1$, for all $c \in \zon$, $\alt(f_k(x\oplus c)) \ge 2^{k-2}$. Since $\sens(f_k) \le k$ by construction of $f_k$, the result follows. 

Proof is by induction on $k$. For $k=2$, $f$ is a function on $3$ variables and it can be 
verified that for all $c \in \zo^3$, $\alt(f(x\oplus c)) \ge 1$. Now consider an $f_{k+1} 
\in \calF$ computed by a decision tree $T$ with the variable $x_t$ as its root. Let $h_1$ 
and $h_2$ be the left and right subtrees of $x_t$ in $T$. Note that $h_1(z')$ and $h_2(z'')$ depends on $n=2^k-1$ variables and belongs to $\calF$ by construction. Hence, by induction, for all 
$c \in \zo^{n}$, $\alt(h_1(z' \oplus c))$ and $\alt(h_2(z'' \oplus c))$ is at least 
$2^{k-2}$. For $m=2^{k+1}-1$, consider any $c = (c',b,c'') \in \zo^m$ where $c',c'' \in \zon$ and $b \in \zo$. Since $h_1$ and $h_2$ are 
variable disjoint, $\alt(f(x\oplus c)) \ge \alt(h_1(z'\oplus c')) + \alt(h_2(z''\oplus c'')) 
\ge 2^{k-2}+2^{k-2} = 2^{k-1}$ completing the induction.
\end{proof}

\noindent
{\bf A family of functions with $\bsens(f) = \Omega(\sens(f)\salt(f))$ :}
Lin and Zhang~\cite{LZ17} showed that for any Boolean function $f:\zon \to \zo$,
\begin{equation} \label{eq:bs:alt:sens}
\bsens(f) = O(\alt(f)^2 \sens(f))
\end{equation} 

	The fact that the measures $\bsens$ and $\sens$ are invariant under shifts implies the following proposition.
\begin{proposition}\label{salt}
For any $f:\zon \to \zo$, $\bsens(f)  \le O(\salt(f)^2 \sens(f))$.
\end{proposition}
\begin{proof}
	For any $b \in \zon$, recall that $f_b(x)$ is defined to be $f(x \oplus b)$.  Applying~\cref{eq:bs:alt:sens} to $f_b$, we get that $\bsens(f_b) = O(\alt(f)^2 \sens(f_b))$. Since, $\bsens$ and $\sens$ are invariant under shifts, for any $b$, $\bsens(f) = \bsens(f_b) = O(\alt(f_b)^2 \sens(f_b)) = O(\alt(f_b)^2 \sens(f))$. Choosing $b$ to be a shift that minimizes the alternation of $f_b$ completes the proof.
\end{proof}
We now exhibit a family of functions for which $\bsens(f)$ is at least $\frac{\sens(f)\cdot \salt(f)}{4}$.

Before proceeding, we show a tight composition result for alternation of Boolean functions when composed with $OR_k$ (which is the $k$ bit Boolean OR function).

For functions $f_1, \ldots, f_k$ where each $f_i: \zon \to
\zo$, define the function $OR_k \circ \overline{f}: \zo^{nk} \to \zo$ as
$\lor_{i=1}^k f_i(x^{(i)})$ where for each $i \in [k]$, $x^{(i)} = (x^{(i)}_1, 
\ldots, x^{(i)}_n) \in \zon$ is input to the function $f_i$.

\begin{lem}\label{alt:comp:or}
Consider $k$ Boolean functions $f_1, \ldots, f_k$ where each $f_i: \zon \to
\zo$ satisfy, $f_i(0^n) = f_i(1^n) = 0$. Then, $$
	\alt( OR_k \circ \overline{f})= \sum_{i=1}^k \alt(f_i). $$
\end{lem}
\begin{proof}
	Let $f = OR_k \circ \overline{f}$ and $\calC$ be a chain in $\zo^{nk}$ for which $
	\alt(f,\calC)$ is maximized. Without loss of generality, let all the functions be non-constant. Let $\calC_i$ be the chain in $\zon$ obtained by 
	restricting  $\calC$ to variables $x_1^{(i)}, \ldots, x_n^{(i)}$ of $f_i$. Observe 
	that if $f$ changes it value, it must be that at least 
	one of the $f_i$'s have changed their evaluation along the chain $\calC$. Since the functions are variable 
	disjoint, such a change must be witnessed in the chain $\calC_i$ for some $i$. Hence 
	$$\alt(f) = \alt(f, \calC) \le \sum_{i=1}^k \alt(f_i, \calC_i) \le \sum_{i=1}^k 
	\alt(f_i) $$
To show that $\alt(f) \ge \sum_{i=1}^k \alt(f_i) $, we exhibit a chain $\calC$ in $
\zo^{nk}$ of alternation $\sum_{i=1}^k \alt(f_i) $. Let $\calC_i = (0^n = z^{(i0)} \prec z^{(i1)} \prec z^{(i2)} \prec \ldots \prec z^{(in)} = 1^n)$ be a chain in $\zon$ for which $f_i$ 
achieves maximum alternation. We construct a chain $\calC$ by ``gluing'' together these 
$k$ chains. More precisely, let $\calC$ by the chain such that for all $i \in [k]$, when 
restricted to the variables $x_{1}^{(i)}, \ldots, x_{n}^{(i)}$, we get a chain given by,
$$ \overbrace{0^n \prec \ldots \prec 0^n}^{n(i-1) \text{ times} }  \prec z^{(i0)} \prec z^{(i1)} \prec z^{(i2)} \prec \ldots \prec z^{(in)} \prec \overbrace{ 1^n \prec \ldots\prec 1^n}^{n(k-i) \text{ times}} 
$$
By construction of $\calC$, since $f_j(0^n) = f_j(1^n) = 0$ for all $j \in [k]$, at any input of the chain $\calC$, there is exactly one $f_i$ that causes $f$ to alternate. Hence, 
$$\alt(f,\calC) \ge \sum_{i=1}^k \alt(f_i, \calC_i) = \sum_{i=1}^k \alt(f_i) $$
\end{proof}

\begin{proposition}
\label{prop:tightness-salt}
There exists a family of Boolean functions for which $\bsens(f) \ge \frac{\sens(f) \cdot \salt(f)}{4}.$
\end{proposition}
\begin{proof}
We consider the Rubinstein's function $f_R:\{0,1\}^{n^2} \to
\{0,1\}$~\cite{Rub95} where the input is treated as $n\times n$ matrix 
which evaluates to $1$ iff there is a row with two consecutive ones starting at
the odd position and rest of the entries being zero. Alternatively, we can view $f_R$ as 
$OR_n \circ \overline{h}$ with $h:\zon \to \zo$ where $h(a) = 1$ iff there are two 
consecutive ones starting at the odd position with rest of the entries as zero in $a \in 
\zon$. It can be verified that $\alt(h) =2$. Since $h(0^n) = h(1^n) = 0$, applying~
\cref{alt:comp:or} with $f_i = h$ for all $i \in [n]$, we get that $\alt(f_R) =  \alt(h)
\cdot n = 2n$. 
It is known that $\bsens(f_R) \ge \frac{n^2}{2}$ while $\sens(f_R) \le n$~\cite{Rub95}, thereby showing that $\bsens(f_R) \ge \frac{\sens(f_R)
\cdot \alt(f_R)}{4} \ge \frac{\sens(f_R)
\cdot \salt(f_R)}{4}$. 
\end{proof}
We remark that the above bound is stronger than what is needed in the context because, $\bsens(f_R) \ge \frac{\sens(f_R)
\cdot \alt(f_R)}{4}$. 

\paragraph{Lower bounding $\salt$ :} By definition, $\salt(f) \le \alt(f)$ and in addition, we have seen a Boolean function $f$ for which $\salt(f) = 1$ while $\alt(f) = \Omega(n)$. This makes $\alt(f)$ particularly unsuitable in obtaining lower bounds on $\salt(f)$. We define a modified variant of the measure alternation called as \emph{subcube alternation} and show that this new measure is always a lower bound on $\salt(f)$.

To define this variant, we define the following notion of restrictions.
For any $S \sse [n]$, define $f|_S$ as the function $f$ defined on the domain $\set{x | x \le e_S}$ and $f|_{\overline{S}}(x)$ as $f(x\oplus e_S)$ for $\set{x | x \ge e_{S}}$.

\begin{defnn}[Subcube alternation]
	For a Boolean function $f$, define the subcube alternation $\scalt$ of $f$ as $\scalt(f)=	\min_{B \sse [n]} (\alt(f|_B) + \alt(f|_{\overline{B}})).$
\end{defnn}

More precisely (in \cref{scalt:lb}), we show that $\forall~f$, $\salt(f) \ge \scalt(f)$ . In arguing the same, we use the following claim which gives an exact expression for maximum alternation of a shifted functions over all chains that contain the shift. 

\begin{lemma} \label{shift:alt:lb}
	For $f:\zon \to \zo$ and any $B \subseteq [n]$, and let $\calC_B$ be the collection of maximal chains containing $e_B$. Then,
	\[\max_{\sigma \in \calC_B} \alt(f(x\oplus e_B), \sigma) = \alt(f|_B) + \alt(f|_{\overline{B}}).\]
\end{lemma}
\begin{proof}
	Let $g(x) = f(x \oplus e_B)$. Denote by $\overline{x}$ the bitwise
	complement of $x$. We claim that,
	\begin{align}
	\forall x : x \le e_B,~ & g(x) = f|_B(\overline{x}) \label{eq1}\\
	\forall x : x \ge e_B,~ & g(x) = f|_{\overline{B}}(x) \label{eq2}
	\end{align}
	\cref{fig:shifting} illustrates the subcubes of interest in the original function and how they change for the function under shift. Now for any chain $\sigma$ containing $e_B$ in the Boolean
	hypercube, $\alt(g,\sigma) = \alt(f|_B) + \alt(f_{\overline{B}})$.
	
	To see \cref{eq2} observe that for any $x \ge e_B$, $x = y \oplus e_B$
	with $y \le e_{\overline{B}}$. Hence $g(x) = f(y) =
	f|_{\overline{B}}(x)$.  For \cref{eq1}, since $x \le e_B$, $g(x) = f(x
	\oplus e_B) = f|_B(\overline{x})$ as restricted to $B$, $x \oplus B$
	complements $x$ (with locations outside $B$ set to $0$). 
	
	\begin{figure}[htp!]
		\centering
		\includegraphics[scale=0.8]{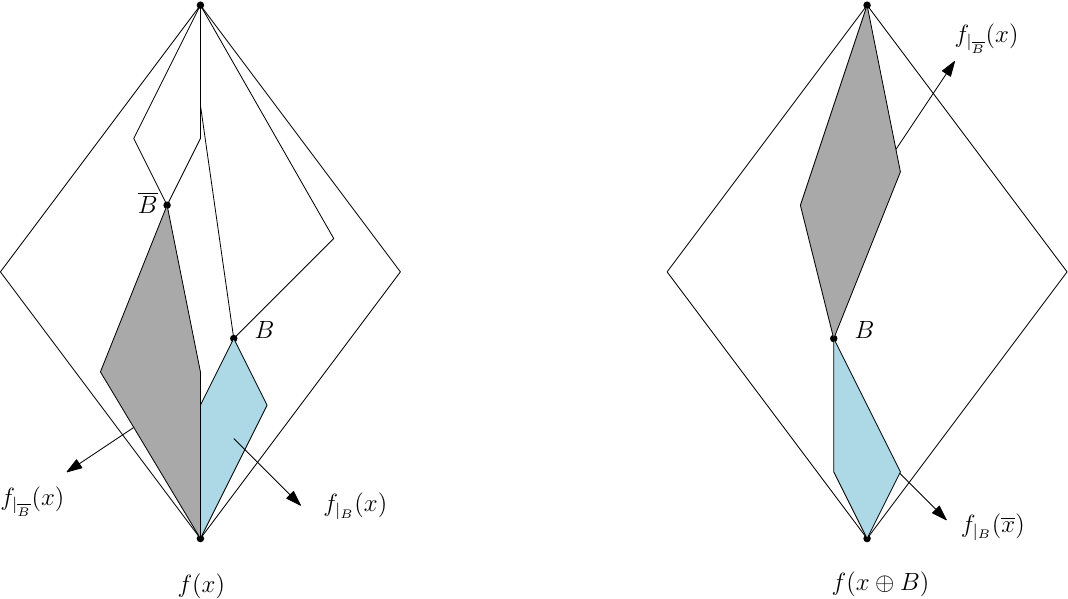}
		\caption{Boolean function $f$ under shift}
		\label{fig:shifting}
	\end{figure}	
	
	Any maximal chain $\sigma$ containing $e_B$ must completely lie in the subcubes $\set{x \mid x \le e_B}$ and $\set{x \mid x \ge e_B}$. Hence, $\max_{\sigma \in \calC_B} \alt(f(x\oplus e_B), \sigma) \le \alt(f|_B) + \alt(f|_{\overline{B}})$. 
	Also, any maximal chain in the subcubes mentioned can be combined in the natural way to get a maximal chain for the whole subcube which contains $e_B$. Hence $\max_{\sigma \in \calC_B} \alt(f(x\oplus e_B), \sigma) \ge \alt(f|_B) + \alt(f|_{\overline{B}})$. 
	
\end{proof}

We can now conclude the lower bound on $\salt$ using~\cref{shift:alt:lb}.
\begin{lemma} \label{scalt:lb}
	For any $f:\zon \to \zo$, 
	$\salt(f) \ge \scalt(f)$. 
	
\end{lemma}
\begin{proof}
	Let $S \sse [n]$ be a shift for which $\alt(f(x\oplus S))$ is minimum and $\calC_S$ denotes the maximal chains containing $e_S$. Hence, 
	\begin{equation}
	\salt(f) = \alt(f(x\oplus e_S)) \ge \max_{\sigma \in \calC_S} \alt(f(x\oplus e_S)) \label{eq:salt:alt}
	\end{equation}
	Combining with \cref{shift:alt:lb}, we have $\salt(f) \ge \alt(f|_S) + \alt(f|_{\overline{S}})$ which is at least $\min_{B \sse [n]} (\alt(f|_B) + \alt(f|_{\overline{B}}))$
\end{proof}

\section{Affine Transforms : Lower Bounds  on Quantum Communication Complexity}

In this section, we study the affine transformation in its full generality applied to block sensitivity and sensitivity,
and use it to prove~\cref{thm:quant:lb} and~\cref{thm:qc-eq:spl} from the introduction. 
We achieve this using  affine transforms as our tool (\cref{bsens:ub:affine}), by which 
we derive a new lower bound for $Q^*_{1/3}(F)$ in terms of $\bsens(f,0^n)$ (\cref{sec:ub:quantum:bs}). 
Using this and a lower bound on $\bsens(f,0^n)$ (\cref{bound:dt:bs:degtwo}), we show that for any Boolean function $f$, and any prime $p$, 
$Q^*_{1/3}(F) \ge \Omega\left (\frac{\sqrt{\DT(f)}}{\deg_p(f)}\right )$. This 
immediately implies that if there is a $p$ such that $\deg_p(f)$ is constant, then $
\D(F) \le 2\DT(f) \le O(Q^*_{1/3}(F)^2)$ thereby answering~\cref{qc-equiv} in 
positive for such functions. We relax this requirement and show that if there exists  
distinct primes $p$ and $q$ for which $\deg_p(f)$ and $\deg_q(f)$ are not very close, then $\D(F) \le 
\poly(Q^*_{1/3}(F))$ (\cref{thm:qc-eq:spl}).

\subsection{Upper Bound for Block Sensitivity via Affine Transforms}\label{bsens:ub:affine}
In this section, we describe our main tool.
Given an $f:\zon \to \zo$ and any $a \in \zon$, we exhibit an affine transform $A:
\F_2^n \to \F_2^n$ such that for $g(x) = f(Ax)$, $\bsens(f,a) \le \sens(g,0^n)$.

Before describing the affine transform, we note that a linear transform is already known to achieve a weaker bound of $\bsens(f) \le O(\sens(g)^2)$ due to Sherstov~\cite{S10}.
\begin{proposition}[Lemma 3.3 of~\cite{S10}]
For any $f:\F_2^n \to \zo$, there exists a linear transform $L:\F_2^n \to \F_2^n$ such that for $g(x) = f(Lx)$, $\bsens(f) = O(\sens(g)^2)$. 

\end{proposition}
See~\cref{lt:sherstov} in~\cref{app:sherstov} for an explicit description of the linear transform achieving the bounds in the above proposition.

Now we describe an affine transform which improves the bound on $\bsens(f)$ in the above proposition to linear in $\sens(g)$. This affine transform has already been used in  Nisan and Szegedy (see Lemma 7 of~\cite{NS94}) 
to show that $\bsens(f) \le 2\deg(f)^2$. Since the exact form of $g$ is relevant in the 
subsequent arguments, we explicitly prove it here bringing out the structure of the affine transform that we require.

\begin{replemma}{lt:bs-s}
For any $f:\F_2^n \to \pmo$ and $a \in \zon$, there exists an affine transform $A:
\F_2^n 	\to \F_2^n$	such that for $g(x) = f(A(x))$, 
\begin{enumerate}[(a)]
\item $\bsens(f,a)\le \sens(g,0^n)$, and \label{bsens:bound}
\item $g(x) = f((x_{i_1}, x_{i_2}, \ldots, x_{i_n})\oplus a)$ where $i_1, \ldots, 
i_n \in [n]$ are not necessarily distinct. \label{structure:g}
\end{enumerate}
\end{replemma}
\begin{proof}
Let $\bsens(f,a)=k$ and $\set{B_1,\ldots, B_k}$ be the sensitive blocks on $a$. Since the blocks are disjoint,  $\set{B_i \mid i\in [k]}$ viewed as vectors over $\F_2^n$ are linearly independent.  Hence, there is a linear transform $L:\F_2^n \to 
\F_2^n$ such that $L(e_i) = B_i$ for $i \in [k]$.\footnote{For completeness of definition 
of $L$, for $i \not \in [k]$, we define $L(e_i) = 0^n$.} Define $A(x) = L(x)\oplus a$. For $g(x) = f(A(x))$, 
\begin{eqnarray*}
 \sens(g,0^n) &=& |\set{i \mid g(0^n) \ne g(0^n \oplus e_i), i \in [n]}| \\
   &=& |\set{i \mid f(a) \ne f(a \oplus L(e_i) ), i \in [n] } | = \bsens(f,a)  
\end{eqnarray*}
which completes the proof of main statement and~\cref{bsens:bound}. \cref{structure:g} holds as the sensitive blocks are disjoint.
\end{proof}

\subsection{From Block Sensitivity Lower Bound at $0^n$ to Quantum Communication Lower Bounds}\label{sec:ub:quantum:bs}

We now prove a lower bound for $Q^*_{1/3}(F)$ in terms of $\bsens(f,0^n)$.
\begin{reptheorem}{ub:quantum:bs}
Let $f:\zon \to \pmo$ and $F(x,y) = f(x \land y)$, then,
$$Q^*_{1/3}(F) = \Omega\left(\sqrt{\bsens(f,0^n)}\right).$$
\end{reptheorem}
\begin{proof}
We first state a weaker version of this result which follows from
Theorem 4.2 of Sherstov~\cite{S10}. The result, which is based on a powerful method of proving quantum communication lower bounds due to Razborov~\cite{R03} and Klauck~\cite{K07}, says that
for a Boolean function $g:\zon \to \pmo$ with $G(x,y) = g(x \land y)$,   if there 
exists an $z \in \zon$ such that $z_i = 0$ for all $i \in [k]$ and $g(z \oplus e_1) = 
g(z \oplus e_2) = \ldots = g(z \oplus e_k) \ne g(z)$, then $Q^*_{1/3}(G) = 
\Omega(\sqrt{k})$. 
This immediately implies that for any $g:\zon \to \pmo$, 
\begin{equation} \label{lb:quantum}
Q^*_{1/3}(G) = \Omega\left(\sqrt{\sens(g,0^n)}\right)
\end{equation}
Given an $f$, we now describe a $g:\zon \to \pmo$ such that 
$Q^*_{1/3}(F) \ge Q^*_{1/3}(G)$ and $Q^*_{1/3}(G) = \Omega(\sqrt{\bsens(f,0^n)})$ as follows thereby completing the proof.

Applying~\cref{lt:bs-s} with $a=0^n$ to $f$, we obtain $g(x) = f(x_{i_1}, x_{i_2}, \ldots, 
x_{i_n})$. We note that $F$ and $G$ can be viewed as a $2^n \times 2^n$ matrix with $(x,y)$th entry being $f(x \land y)$ and $g(x \land y)$ respectively. By construction of $g$, using the observation that the matrix $G$ appears as a submatrix of $F$, $Q^*_{1/3}(F) \ge Q^*_{1/3}(G)$. 
This observation is used in Sherstov (for instance, see proof of Theorem 5.1 of~\cite{S10}) without giving details. For completeness, we give the details here. Let $S = \set{i_1, \ldots i_n } \sse [n]$ of size $k$. For $j \in S$, let $B_j = \set{t \mid 
i_t=j}$. Hence $g$ depends only on these $k$ input variables of $S$ and all the 
variables with indices in $B_j$ are assigned the variable $x_j$. This implies that 
\begin{equation} \label{eq:gfrel}
g(x) = f(\oplus_{j \in S} x_je_{B_j})
\end{equation}

We now exhibit a submatrix of $F$ containing $G$. Consider the submatrix of $F$ with rows and columns restricted to 
$$W = \set{a_1e_{B_1}  \oplus a_2e_{B_2} \oplus \ldots  a_k e_{B_k} \mid  (a_1,a_2\ldots,a_k) \in \zo^k}.$$ 
For $u,y \in W$, 
\begin{align*}
F(u,y) & = f( u\land y) \\  & = f((u_1e_{B_1} \oplus \ldots \oplus   u_ke_{B_k}) \land (y_1e_{B_1} \oplus \ldots \oplus   y_ke_{B_k}) ) \\
& = f(u_1 \land y_1 e_{B_1} \oplus \ldots \oplus u_k\land y_ke_{B_k}) &&[\text{$B_j$s are disjoint}] \\
& = g(u \land y) && [\text{By~\cref{eq:gfrel}}]
\end{align*}

Applying~\cref{lb:quantum} to the $g$ obtained, we have 
$Q^*_{1/3}(G)  \ge \Omega(\sqrt{\sens(g,0^n)})$. Hence, by~\cref{bsens:bound} of~\cref{lt:bs-s}, as $a=0^n$, we have 
$Q^*_{1/3}(G)  \ge \Omega(\sqrt{\bsens(f,0^n)})$.
\end{proof}

\begin{rem}\label{rem:quantum}
Observe that for an arbitrary $a \in \zon$ for $g(x) = f(x \oplus a)$, the statement $Q^*_{1/3}(G) \le Q^*_{1/3}(F)$ does not hold. Otherwise, we would have $Q^*_{1/3}(F) = \Omega(\sqrt{\bsens(f)})$ for all $f$ which is not true (see the discussion after \cref{ub:quantum:bs} in the Introduction).
\end{rem}

\subsection{Putting Them Together}
We are now ready to prove~\cref{thm:quant:lb} and~\cref{thm:qc-eq:spl}.
A critical component of our proof is the following stronger connection between $\DT(f)$ and $\bsens(f,0^n)$. Buhrman and de Wolf, in their survey~\cite{BW02},  showed the following with the proof attributed to Noam Nisan and Roman Smolensky. 
\begin{lemma}[\cite{BW02}] For any Boolean function $f:\zon \to \zo$,
	$\DT(f) \le \bsens(f) \cdot \deg(f)^2$
\end{lemma}

The same proof can be adapted 
to show the following 
strengthening of their result. 

\begin{proposition}\label{bound:dt:bs:degtwo}
For any $f:\zon \to \zo$, and any prime $p$, $$\DT(f) \le \bsens(f,0^n) \cdot \deg_p(f)^2.$$
\end{proposition}
\begin{proof}
We observe that the arguments of Buhrman and de Wolf (more specifically, Lemma 5, Lemma 6 and Theorem 12 of~\cite{BW02}), can give a stronger upper bound than $\bsens(f)\cdot  \deg(f)^2$, namely $\bsens(f,0^n)\cdot \deg_p(f)^2$. 
This is important in our context since we are able to bound $Q^*_{1/3}(F)$ only by $\bsens(f,0^n)$.

Let $p_f(x) \in \F_p[x_1,\ldots,x_n]$ be an $\F_p$ polynomial representation of $f
$. As $p_f$ is a multilinear, we view monomials as subsets of variables. We define \textit{size} of a monomial as the number of variables in it.  Let $S_f$ be the collection of all monomials of maximal size in $p_f$.  We show that,
\begin{claim}\label{claim}
 For any Boolean function $f$, there is a set of variables of size at most $\bsens(f,0^n) \cdot \deg_p(f)$ 
which has a non-empty intersection with all the monomials in $S_f$. 
\end{claim}
We call this set as a hitting set for $S_f$. We now assume this claim. Hence, querying these 
variables fixes them and results in a function whose $\F_p$-degree is at most $\deg_p(f)-1$. We 
repeat this on the resulting function to obtain the desired decision tree where at 
most $\bsens(f,0^n)\cdot \deg_p(f)^2$ variables gets queried. 

\paragraph{Proof of~\cref{claim}}We now argue the existence of a hitting set, which has a non-empty intersection with all the monomials in $S_f$, of size at most $\bsens(f,0^n) \cdot \deg_p(f)$.

Firstly, observe that every monomial $m$ in $S_f$ must have a non-empty set  $B$ of indices of variables in $m$ such that $f(0^n) \ne f(0^n \oplus e_B)$. To see this, restrict $f$ to indices in the monomial $m$ by setting all variables not in the monomial to $0$. Let $g$ be the resulting function. By construction, $g$ is non-constant as the monomial $m$ appears in the $\F_p$ representation of $g$. Hence there must be some setting of the input to $g$ such that its evaluation differs from that of the all zero input.

We construct a hitting set $H$ as follows: for each monomial $m$ in $S_f$, if no variable in $H$ appear in $m$, add all the variables in it to $H$. Since, each such monomial contains a sensitive block on the input $0^n$, the number of monomials that gets added to $H$ is at most $\bsens(f,0^n)$. Since each monomial is of size at most $\deg_p(f)$, total size of the hitting set is at most $\bsens(f,0^n)\cdot \deg_p(f)$.
\end{proof}

We now give a proof of~\cref{thm:quant:lb} and~\cref{thm:qc-eq:spl}.
\begin{reptheorem}{thm:quant:lb}
Fix a prime $p$. Let $f:\zon \to \pmo$ where $f$ depends on all the inputs. Let $F(x,y) = f(x \land y)$. For any $0 < \epsilon < 1$  such that $\deg_p(f) \le (1-\epsilon)\log n$, we have  $$Q^*_{1/3}(F) = \Omega\left (\frac{n^{\epsilon/2}}{\log n} \right ).$$ 
\end{reptheorem}
\begin{proof}
Applying~\cref{ub:quantum:bs} and \cref{bound:dt:bs:degtwo}, we have
\begin{equation} \label{eq:quantum-lb}
 Q^*_{1/3}(F) \ge \Omega\left (\frac{\sqrt{\DT(f)}}{\deg_p(f)}\right )
\end{equation} 
  As observed in Gopalan \etal~\cite{GLS09}, by a modification to an argument in the proof of Nisan and Szegedy (Theorem 1 of~\cite{NS94}), it can be shown that  $\deg(f) \ge \frac{n}{2^{\deg_p(f)}}$. Since, $\DT(f) \ge \deg(f)$, we have $\DT(f)  \ge \frac{n}{2^{\deg_p(f)}}$. Hence ~\cref{eq:quantum-lb} gives,
\begin{align*}
Q^*_{1/3}(F) & = \Omega \left (\frac{\sqrt{n}}{\deg_p(f)2^{\deg_p(f)/2}} \right )  = \Omega \left (\frac{n^{\epsilon/2}}{(1-\epsilon)\log n} \right ) 
\end{align*}
where the last lower bound follows upon applying the bound on $\deg_p(f)$.
\end{proof}

As a demonstrative example, we show a weaker lower bound on quantum communication complexity with prior entanglement for the generalized inner product function $\GIP_{n,k}(x,y) \defn \oplus_{i=1}^n \bigwedge_{j=1}^k (x_{ij} \land y_{ij})$ when $k= \frac{1}{2} \log n$. We remark that a lower bound of $\Omega(n)$ is known for the inner product function~\cite{RWMA99}.

Note that $\GIP_{n,k}$ can be expressed as $f\circ \land$, where $f(z) \defn \oplus_{i=1}^n \bigwedge_{j=1}^k z_{ij}$, with $\degtwo(f) = k$. Applying~\cref{thm:quant:lb} with $\epsilon = 1/2$ and $p=2$, we have $Q^*_{1/3}(\GIP_{n,\frac{1}{2}\log n}) = \Omega\left (\frac{n^{1/4}}{\log n} \right )$. Though this bound is arguably weak,~\cref{thm:quant:lb} gives a non-trivial lower bound for a all those Boolean functions $f$ with small $\deg_p(f)$ for some prime $p$.

\begin{reptheorem}{thm:qc-eq:spl}
Let $f:\zon \to \pmo$ with $F(x,y) = f(x \land y)$. 
Fix $0 < \epsilon < 1$. If there exists distinct primes $p$, $q$ such that $\deg_q(f) = \Omega( \deg_p(f)^{\frac{2}{1-\epsilon}})$, then $\D(F) = O(Q^*_{1/3}(F)^{2/\epsilon})$.
\end{reptheorem}
\begin{proof}
Applying, \cref{ub:quantum:bs} and \cref{bound:dt:bs:degtwo}, for any prime $t$, $Q^*_{1/3}(F) \ge \Omega\left (\frac{\sqrt{\DT(f)}}{\deg_t(f)}\right )$. By hypothesis, $\deg_p(f) \le O(\deg_q(f)^{\frac{1-\epsilon}{2}}) \le O(\DT(f)^{\frac{1-\epsilon}{2}})$ implying that for $t=p$, $\D(F) \le 2\DT(f) \le O(Q^*_{1/3}(F)^{2/\epsilon})$. 
\end{proof}

\begin{rem}
	For any Boolean function $f$, if there exists a prime $p$ with $\deg_p(f) \le c\log n$ for some $c<1/2$, then by main result of~\cite{GLS09} relating degree of Boolean functions under different field characteristics, for any prime $q\ne p$, $\deg_q(f) = \Omega(\frac{n^{1-2c}}{c\log p \log n}) = \Omega((\log n)^2)$. Hence any such $f$ satisfies the condition that $\deg_q(f) = \Omega( \deg_p(f)^{\frac{2}{1-\epsilon}})$ for some constant $\epsilon$ and by~\cref{thm:qc-eq:spl}, $\D(F) = O(Q^*_{1/3}(F)^{2/\epsilon})$.
\end{rem}

\section{Linear Transforms : Sensitivity versus Sparsity}
\label{sec:sens-sps}
Continuing in the theme of affine transforms, in this section,
we first establish an upper bound on alternation of a function in terms of  sensitivity of the function after 
application of a suitable linear transform.
Using this, we show the existence of a function whose sensitivity is 
asymptotically as large as square root of sparsity (see introduction for a motivation and  discussion).
\begin{replemma}{alt-sens-lt}
	For any $f:\zon \to \zo$, there exists an invertible linear transform $L:\F_2^n 
	\to \F_2^n$
	such that for $g(x) = f(L(x))$, $\alt(f) \le 2\sens(g)+1$.
\end{replemma}
\begin{proof}
	Let $0^n \prec x_1 \prec x_2 \ldots \prec x_n = 1^n$ be a chain
	$\calC$ of maximum alternation in the Boolean hypercube of $f$.
	Since chain $\calC$ has maximum alternation, there must be at least
	$(\alt(f)-1)/2$ many zeros and $(\alt(f)-1)/2$ many ones when the $x_i$s are
	evaluated on $f$.  Note that the set of $n$ distinct inputs $x_1,x_2,\ldots,
	x_n$ seen as vectors in $\F_2^n$ are linearly independent and hence
	is a basis of $\F_2^n$. Hence there exists an
	invertible\footnote{$L$ is actually the change of basis transform from
	standard basis vectors to $x_i$s and hence is bijective.} linear transform
	$L:\F_2^n \to \F_2^n$ taking standard basis vectors to the these
	vectors, i,e. $L(e_i) = x_i$ for $i \in [n]$. 

	To prove the result, we now show that $\sens(g,0^n) \ge
	\frac{\alt(f)-1}{2}$. The neighbors of $0^n$ in the hypercube of $g$
	are $\set{e_i \mid i \in [n]}$ and each of them evaluates to $g(e_i) =
	f(L(e_i)) = f(x_i)$ for $i \in [n]$.  Since there are at least
	$(\alt(f)-1)/2$ many zero and at least those many ones among $x_i$s when
	evaluated by $f$, there must be at least $(\alt(f)-1)/2$ many neighbors of
	$0^n$ which differ in evaluation with $g(0^n)$ (independent of the value of
	$g(0^n)$). Hence $\sens(g) \ge s(g,0^n) \ge
	\frac{\alt(f)-1}{2}$ which completes the proof.
\end{proof}
We now describe the family of functions and argue an exponential gap between sensitivity and logarithm of sparsity, as stated in the following 
Theorem.
\begin{reptheorem}{gap:salt:degtwo}
There exists a family of functions $\set{g_k\mid k \in \N}$ such that $$\sens(g_k) 
\ge \frac{\sqrt{\sps(g_k)}}{2}-1.$$
\end{reptheorem}
\begin{proof}
 For the family of functions $f_k \in \calF$ (\cref{def:fun}), $\alt(f_k) \ge 2^{(\log \sps(f_k))/2}-1$~\cite{JDS18}. 

We now use this family $\calF$ to describe the family of functions $g_k$. For every $f_k \in \calF$, let 
$g_k(x) = f_k(L(x))$ such that $\alt(f_k) \le 2\sens(g_k)+1 $ as guaranteed by~
\cref{alt-sens-lt}. Since, we have $\alt(f_k) \ge 2^{(\log \sps(f_k))/2}-1$, it must be that 
$$
\sens(g_k)   \ge \frac{1}{2} (\alt(f_k)-1)  \ge \frac{1}{2} (2^{(\log \sps(f_k))/2}-2) \ge \frac{\sqrt{\sps(f_k)}}{2}-1
$$
As the parameter $\sps$ does not change under invertible linear transforms (Ex 3.1 \cite{AOBF}), $
\sens(g_k) \ge 0.5\sqrt{\sps(f_k)}-1 = 0.5\sqrt{\sps(g_k)}-1$.
\end{proof}
We now describe how the family of Boolean functions in~\cref{gap:salt:degtwo} rule out a possibility of settling XOR Log-Rank conjecture, a conjecture in classical communication complexity, using a recent proof of Sensitivity Conjecture. First, we describe the  XOR Log-Rank conjecture and then give a potential way to prove the XOR Log-Rank conjecture using the recent resolution of Sensitivity Conjecture~\cite{H19}. Following this, we argue how the family of Boolean functions in~\cref{gap:salt:degtwo} rules out this possibility.

For an $f:\zon \to \zo$, define $F_\oplus:\zon \times \zon \to \zo$ as $F_\oplus(x,y) = f(x\oplus y)$. The XOR Log-Rank conjecture says that, for every $f$, the deterministic communication cost of computing the corresponding $F_\oplus$ must satisfy $\D(F_\oplus) = \poly(\log \sps(f))$. An equivalent formulation of the Sensitivity Conjecture due to Hatami \etal (Proposition 5.10, \cite{HKP11}) says that for every $f$, $\D(F_\oplus) = \poly(\sens(f))$. With the Sensitivity conjecture now proven~\cite{H19}, one way to prove the XOR Log-Rank conjecture is to show that for all Boolean functions $f$, $\sens(f) \le \poly(\log \sps(f))$. Unfortunately, the existence of a family of Boolean functions in~\cref{gap:salt:degtwo} rules out this possibility.

\section{Conclusion and Future directions}
In this paper, we study the Boolean function complexity measures, namely sensitivity, block 
sensitivity, and alternation under affine transforms. We showed design of 
special transforms which achieves structurally revealing statements about the 
resulting function. We used their properties to show lower bounds on the bounded 
error quantum communication complexity of Boolean function whose $\F_p$-degree is 
small. We  showed that classical and quantum communication complexity are 
polynomially related for certain special class of functions. We also demonstrated 
Boolean functions where sensitivity of the function is as large as the square root 
of its sparsity.

The main open question is to see if the tools developed here can be pushed to remove the restriction on $\deg_p$ and $\deg_q$ of Boolean functions in~\cref{thm:qc-eq:spl} thereby proving the Quantum Classical equivalence~(\cref{qc-equiv}).
\section{Acknowledgment}
The authors would like to thank the anonymous reviewers for their constructive comments to this paper, specifically for pointing out an error in the earlier version of~\cref{ub:quantum:bs} by giving examples. 
See the \cref{rem:quantum} and the discussion after~\cref{ub:quantum:bs}  of this paper.

\bibliographystyle{alpha}
\bibliography{references}

\begin{appendices}
\section{Quantum communication lower bound from block sensitivity}\label {app:sherstov}
	Sherstov in~\cite{S10} showed the following lower bound on quantum communication cost of an affine shift of a Boolean function in terms of its block sensitivity.
	
	\begin{corollary}[Corollary 4.5 of~\cite{S10}]
		Let $f:\zon \to \pmo$ be given. Then for some $z \in \zon$, the matrix $F'= 
		[f_z(x \land y)]_{x,y} = [f(\ldots, (x_i \land y_i) \oplus z_i, \ldots)]_{x,y}$ 
		obeys$$ Q^*_{1/3}(F') = \Omega(\sqrt{\bsens(f)})$$
	\end{corollary}
	
	In this section, we elaborate on why one
	
	cannot set $z=0^n$ for all Boolean functions and obtain~\cref{ub:quantum:bs}. The 
	above corollary crucially uses two results. The first one is Lemma 3.3 of~
	\cite{S10} which shows that there exists a Boolean function $g:\zon \to \zo$ such 
	that $\bsens(f) \le O(\sens(g)^2)$  which is similar in spirit to~\cref{lt:bs-s}. 
	The second one is Theorem 4.2 of~\cite{S10} which shows a lower bound for 
	$Q^*_{1/3}(G)$ in terms of sensitivity of $g$ (where $G(x,y) = g(x \land y)$). We 
	reproduce the respective statements of both below.

	\begin{lemma}[Lemma 3.3 of~\cite{S10}]
		Let $f:\zon \to \pmo$. Then there exists a $g:\zon \to \pmo$ such that $s(g) = 
		\Omega( \sqrt{\bsens(f)})$  and $g(x) = f(x_{i_1},\ldots, x_{i_n})$ for some 
		$i_1,\ldots,i_n \in [n]$
	\end{lemma}
	
	The function $g$ is defined as follows.
	
	Let $z$ be the input on which $\bsens(f,z)$ is maximum and $f(z)=0$. Let $S_1,\ldots, S_k \sse [n]$ be the sensitive blocks on $z$. Define $A_i = \set{j \in S_i \mid z_j =0}$ and $B_i =  \set{j \in S_i \mid z_j =1}$. Let $I$ be the indices $i \in [k]$ such that both $A_i$ and $B_i$ are both non-empty.
	
	Then 
	$$ g(x) = f\left ( \bigoplus_{i \in I} x_{min A_i} e_{A_i} \oplus \bigoplus_{i \in I} x_{min B_i} e_{B_i}  \oplus \bigoplus_{i \in [k]\setminus I} x_{min S_i} e_{S_i} \oplus  \bigoplus_{i \not \in S_1 \cup \ldots \cup S_k }  x_ie_i \right )$$

	\begin{observation}\label{lt:sherstov}
		We observe that the above result of Sherstov (Lemma 3.3 of~\cite{S10}) can be seen as applying a suitable linear transform to the Boolean function $f$ to bound the block sensitivity of $f$ which is similar in spirit to~\cref{lt:bs-s}.
		
		More precisely, the $g$ obtained in Lemma 3.3 of~\cite{S10} can be described as $f(L(x))$ where $L$ is defined as, for $j \in [n]$,
		\[ L(e_j) = \begin{cases}
		e_j \text{ if } j \not \in  S_1 \cup \ldots \cup S_k \\
		e_{A_i} \text{ if } \exists i\in [k], \text{ such that }  j = \min \{ A_i \} \\
		e_{B_i} \text{ if } \exists i\in [k], \text{ such that }  j = \min \{ B_i \} \\
		0^n \text{ otherwise }
		\end{cases}
		\]
		
	\end{observation}

	By definition $g$ as above, Sherstov showed that $\sens(g,z) = \Omega(\sqrt{\bsens(f)})$. 
	
	\begin{theorem}[Theorem 4.2 of~\cite{S10}]
		For a Boolean function $g:\zon \to \pmo$ with $G(x,y) = g(x \land y)$,   if there 
		exists an $w \in \zon$ such that $w_i = 0$ for $i \in [k]$ and $g(w \oplus e_1) = 
		g(w \oplus e_2) = \ldots = g(w \oplus e_k) \ne g(w)$, then $Q^*_{1/3}(G) = 
		\Omega(\sqrt{k})$.
	\end{theorem}
	
	To use the above result, one way is to start with a function $g$ for which sensitivity is large at $0^n$. To achieve, consider the shifted function $f_z$ where $z$ is the same input on which block sensitivity is maximized as before. This is because, by the choice of $z$, $f_z$ will have maximum block sensitivity at $0^n$ which upon applying Lemma 3.3 of~\cite{S10} ensures that the function $g$ obtained has a large $k$ (\ie~sensitivity) at $0^n$. This is exactly what is achieved in the proof of Corollary 4.5 of~\cite{S10}.
	
	Hence the choice is $z$ is tied up with the block sensitivity of function $f$.
\end{appendices}

\end{document}